\documentclass[conference,a4paper]{IEEEtran}
\usepackage{tikz}
\usetikzlibrary{matrix}
\usepackage{graphicx}
\usepackage{epstopdf}
\usepackage{wrapfig}
\usepackage{amsfonts}
\usepackage{makecell}
\usepackage{pdfpages}
\usepackage[utf8]{inputenc} 
\usepackage{url}
\usepackage{ifthen}
\usepackage{cite}
\usepackage[cmex10]{amsmath}
\usepackage{mathtools}
\usepackage{amsthm}
\usetikzlibrary{shapes,arrows}
\usepackage{amssymb}
\usepackage{mathrsfs}

\usepackage[T1]{fontenc} 
\usepackage{mleftright}       
\mleftright                   

\usepackage{graphicx}         
\usepackage{booktabs}         

\usepackage{stfloats}
\usepackage{float}

\usetikzlibrary{shapes,arrows}

\DeclareMathOperator*{\argmin}{arg\,min}

\newtheorem{theorem}{Theorem}

\newtheorem{definition}{Definition}

\begin{document}

\title{An Achievability Bound for Variable-Length Stop-Feedback Coding over the Gaussian Channel} 
\author{%
  \IEEEauthorblockN{Ioannis Papoutsidakis, Robert J. Piechocki, and Angela Doufexi}\\
  \IEEEauthorblockA{Communication Systems and Networks Group\\ 
  					Department of Electrical and Electronic Engineering\\
                    University of Bristol\\ 
                    Bristol, BS8 1UB, UK\\
                    Email: \{ioannis.papoutsidakis, r.j.piechocki, a.doufexi\}@bristol.ac.uk}
}

\maketitle

\begin{abstract}Feedback holds a pivotal role in practical communication schemes, even though it does not enhance channel capacity. Its main attribute includes adaptability in transmission that allows for a higher rate of convergence of the error probability to zero with respect to blocklength. Motivated by this fact, we present a non-asymptotic achievability bound for variable-length coding with stop-feedback. Specifically, a general achievability bound is derived, that employs a random coding ensemble in combination with minimum distance decoding. The general bound is particularized for the Gaussian channel. Numerical evaluation of the bound confirms the significant value of feedback compared to transmission with fixed blocklength coding and without feedback.
\end{abstract}

\section{Introduction}

Non-asymptotic upper and lower bounds on the achievable rate of noisy channels are the focus of research related to low-latency communications. Communication schemes that utilize feedback, such as automatic repeat request (ARQ) and hybrid automatic repeat request (HARQ), are widely used in practical applications and support high reliability and adaptation to the channel conditions. The non-asymptotic performance of coding strategies with feedback is the theme of this paper.

In general, a fundamental result of information theory shows that feedback does not increase the capacity of a noisy channel. Despite this fact, feedback is quite useful since it can increase the zero-error capacity, as shown in \cite{1056798}, and can also increase the rate at which the probability of error converges to zero as a function of the blocklength, as discussed in \cite{Bur76}. Some of the most recent significant results for fixed blocklength codes without feedback can be found in \cite{poly}. Similar results were later produced for discrete channels with noiseless and instantaneous feedback in \cite{5961844}. In \cite{9204378}, noisy stop feedback is considered for the finite blocklength regime, with promising results for the binary-input Gaussian channel and the quasi-static Rayleigh fading channel. Specifically, in the context of the fading channel, feedback provides increased diversity as well as better achievable rates compared to fixed blocklength codes.  Asymptotic approximations on the average blocklength of variable-length feedback codes are derived in \cite{8259002} for the Gaussian point-to-point and multiple access channels. 

Variable-length coding with stop-feedback (VLSF) is the general class of widely used techniques, such as HARQ with incremental redundancy. In these communication schemes, transmission is terminated by a binary feedback signal when the decoder has sufficient channel outputs to determine the message. This signal is usually known as an acknowledgement in case of ``stop'' or a negative acknowledgement in case of ``continue''. A non-asymptotic achievability bound of VLSF coding can provide information not only for the average blocklength but also the distribution of the decoding time. This characteristic is not captured by asymptotic approximations such as those in \cite{8259002}. Additionally, asymptotic approximations require the heuristic substitution of big $O$ terms, e.g. \cite[(296)]{poly}, and can be inaccurate for small blocklengths.

This work provides an achievability bound on the rate of VLSF codes over the Gaussian channel with an average codebook power constraint. Specifically, it presents a general bound based on the minimum distance decoder and the definition of VLSF codes given by Polyanskiy et al. in \cite{5961844}. The general bound is specialized for the Gaussian channel by employing a setting similar to that in \cite{10206565} for the minimum distance decoder. Finally, the bound is numerically evaluated, enabling a comparison between VLSF coding and fixed blocklength coding over the Gaussian channel.

In section \ref{notation}, the definition of VLSF codes is presented as well as the notation we follow throughout this paper. The main theorem is given in section \ref{general_result}. The particularization of the main result for the Gaussian channel is derived in section \ref{gaus_res}. The final remarks and conclusions are made in section \ref{concl}.

\section{Notation and Definitions}
\label{notation}

Throughout this paper, the blocklength is denoted by $n$. Uppercase letters represent random variables, and lowercase letters denote their realizations. Random vectors are indicated by bold uppercase letters. Superscripts refer to the first $m$ entries of a vector, denoted as $\textbf{X}^m = (X_1, ..., X_m)$. The Euclidean norm is symbolized by $\Vert \cdot \Vert$. The multivariate normal distribution is denoted with  $\mathcal{N}(\mu,\Sigma)$ where $\mu$ is the mean vector and $\Sigma$ is the covariance matrix. The $m\times m$ identity matrix is denoted with $\textbf{I}_m$. The probability density function (pdf) of the non-central chi-squared distribution with $\kappa$ degrees of freedom and non-centrality parameter $\nu$ is represented with $f_{\mathcal{X}_\kappa^2}(x;\nu)$. Equivalently, the pdf of the gamma distribution with shape parameter $\kappa$ and scale parameter $\theta$ is denoted with $f_\Gamma(x;\kappa,\theta)$. Cumulative distribution functions (cdf) are defined analogously but with an upper case function name, e.g. $F_\Gamma(x;\kappa,\theta)$.

Let us consider a sequence of identical conditional probability kernels $\{P_{Y_i|X_i}\}_{i=1}^\infty$ that define a memoryless channel with input alphabet $\mathcal{A}$ and output alphabet $\mathcal{B}$. The following definition of VLSF codes is a modification of the definition of variable length feedback codes given in \cite{5961844}.

%

\begin{definition}
An $(l,M,\epsilon)$ variable-length stop-feedback code, where $l$ is a positive real number, $M$ is a positive integer number, and $0 \leq \epsilon \leq 1$, is defined by:

\begin{enumerate}
\item A space $\mathcal{U}$ with a probability distribution $P_U$ that defines a random variable $U$ which initializes the codebook on both the transmitter and the receiver before the start of transmission.

\item A sequence of encoders $f_n:\mathcal{U}\times\{1,...,M\} \rightarrow \mathcal{A}$, defining channel inputs
\begin{align}
\textbf{X}_n=f_n(U,W)
\end{align}
where $W\in\{1,...,M\}$ is the equiprobable message.

\item A sequence of decoders $g_n:\mathcal{U}\times\mathcal{B}^{n} \rightarrow \{1,...,M\}$ providing the best estimate of $W$ at time $n$.

\item A non-negative integer-valued random variable $\tau$, which represents a stopping time of the filtration $\mathscr{G}_n=\sigma\{U,Y_1,...,Y_n\}$ and satisfies 
\begin{align}
\mathbb{E}[\tau]\leq l.
\end{align}
\end{enumerate}
The final estimation of message $W$ is given at time $\tau$
\begin{align}
\hat{W}=g_{\tau}(U,\textbf{Y}^{\tau})
\end{align}
and satisfies the probability of error constraint
\begin{align}
P(\hat{W}\neq W) \leq \epsilon.
\end{align}
\label{defVLSF}
\end{definition} 
The key element of this definition is the random variable $U$ that is used to define an ensemble of random codes. This is necessary because the existence of one specific code that achieves the constraints of the coding ensemble is not guaranteed, contrary to the fixed-length coding case. 

\section{General Achievability Bound}
The following theorem is the general result of this paper and it describes the constraints that an $(l,M,\epsilon)$ VLSF code must satisfy when a minimum distance decoder is utilized.
\label{general_result}
\begin{theorem}
Fix a threshold $0 \leq \epsilon \leq 1$ on the probability of error and a memoryless channel $\{P_{Y_i|X_i}\}_{i=1}^\infty$. Let the arbitrary mutually independent and identical processes $\textbf{X}_i= (X_{i1},X_{i2},...,X_{in},...)$ for $i=1,...,M$ where $X_{in} \in\mathcal{A}$. Let the marginal distribution of $\textbf{X}_i^n$ be denoted as $P_{\textbf{X}^n}$. The channel output $\textbf{Y}$ is received when $\textbf{X}_1$ is the channel input. Define a sequence of distance functions $d:\mathcal{A}^n\times\mathcal{B}^n\rightarrow\mathbb{R}_{\geq 0}$.  Finally, let the random variables

\begin{align}
V_{1,n} &= d(\textbf{X}_1^n,\textbf{Y}^n),\\
V_{2,n} &= \min_{j=2,...,M } d(\textbf{X}_j^n,\textbf{Y}^n),\\
K_n &= \min_{i=1,2}(V_{i,n}),
\end{align}
and a stopping time
\begin{align}
\tau = \inf\{n \geq 0:\lambda(K_n,\textbf{Y}^n)\leq \epsilon \}
\end{align}
where
\begin{align}
\begin{split}
\lambda(\upsilon,\textbf{y}^n)\geq P(V_{2,n} = \upsilon|K_n =\upsilon,\textbf{Y}^n=\textbf{y}^n).
\end{split}
\end{align}
Then for any $M$ there exists an $(l,M,\epsilon)$ VLSF code with
\begin{align}
\begin{split}
l &\leq \mathbb{E}[\tau].
\end{split}
\end{align}
\label{theorG}
\end{theorem}

\begin{proof}
Based on Definition \ref{defVLSF}, to define a code we need to specify $(U,f_n,g_n,\tau)$. Random variable $U$ is defined as follows.

\begin{align}
\mathcal{U} \triangleq \underbrace{\mathcal{A}^\infty\times ...\times \mathcal{A}^\infty}_\text{$M$ times},
\end{align}
\begin{align}
P_U \triangleq \underbrace{P_{\textbf{X}^\infty} \times ...\times P_{\textbf{X}^\infty}}_\text{$M$ times}
\end{align}
where $P_{\textbf{X}^\infty}$ is the marginal distribution of $\textbf{X}_i^\infty$ for any $i=1,...,M$.
The common information provided to the encoder and the decoder by the realization of $U$ defines $M$ codewords of infinite length $\textbf{C}_i\in \mathcal{A}^\infty$ for $i=1,...,M$. The encoder $f_n$ carries out the mapping
\begin{align}
f_n(w)=C_{wn}
\end{align}
where $C_{wn}$ denotes the $n$th entry of $\textbf{C}_w$.

At the time of the reception of the $n$th symbol, the decoder computes $M$ distances
\begin{align}
S_{j,n} = d(\textbf{C}_j^n,\textbf{Y}^n)
\end{align} 
for $j =1,...,M$. The minimum distance decoder estimates the message $w$ as
\begin{align}
g_n(\textbf{Y}^n)=\argmin_{j=1,...,M} S_{j,n}.
\end{align}

An alternative decoding scheme with higher error rate and average blocklength is the following. Assume, without any loss of generality, that the transmitted message is $W=1$. Then
\begin{align}
g_n(\textbf{Y}^n)=\argmin_{j=1,2} V_{j,n}.
\end{align} 
where
\begin{align}
V_{1,n} = S_{1,n}
\end{align}
and
\begin{align}
V_{2,n} = \min_{j=2,...,M} S_{j,n}.
\end{align}

This conversion is easier to computationally manipulate since it avoids operations with $M$ terms without significantly relaxing the final result. In fact, in the case of a continuous channel, where the $V_{1,n}$ and $V_{2,n}$ are continuous, this conversion does not relax the achievability bound because 
\begin{align}
P(V_{1,n} =V_{2,n})=0.
\label{cntvar}
\end{align}

The probability of error of the minimum distance decoder after $n$ channel uses is bounded as follows. Let
\begin{align}
K_n = \min_{i=1,2}(V_{i,n}) 
\end{align}
and define a function $\lambda$ such that
\begin{align}
\begin{split}
P(&\hat{W} \neq W|\textbf{Y}^n=\textbf{y}^n) \\
&\stackrel{(a)}{\leq} P(V_{2,n} = \upsilon|K_n =\upsilon,\textbf{Y}^n=\textbf{y}^n) \\
&\leq\lambda(\upsilon,\textbf{y}^n)
\end{split}
\end{align}
where $(a)$ is due to the process of minimum distance decoding and the assumption that the event $V_{1,n}=V_{2,n}$ always results in an error. In the case of a continuous channel, the probability of this event is zero as in (\ref{cntvar}), and $(a)$ can be an equality.

Note that the transmission should only stop when,
\begin{align}
P(\hat{W}\neq W|\textbf{Y}^n=\textbf{y}^n) \leq \epsilon.
\label{errconstr}
\end{align}
If the stoppage criterion is $\lambda(\upsilon,\textbf{y}^n) \leq \epsilon$, constraint (\ref{errconstr}) will always be satisfied. Hence, let a stopping time 
\begin{align}
\tau = \inf\{n \geq 0:\lambda(K_n,\textbf{Y}^n)\leq \epsilon \}
\end{align}
then
\begin{align}
l \leq \mathbb{E}[\tau].
\end{align}
\end{proof}

This achievability bound can be used with any channel for which the minimum distance decoder is optimal, such as the binary symmetric channel (BSC), the binary erasure channel (BEC), and the Gaussian channel. The focus of this work is on the latter. We conjecture, however, that in the cases where information density is maximized by the minimization of distance, such as in the BSC and the BEC \cite{poly}, the particularization of the bound is equivalent to the one in \cite{5961844}.  

\section{The Gaussian Channel}
\label{gaus_res}
The Gaussian Channel, a continuous channel with additive Gaussian noise and a power constraint, is the most fundamental and well-studied among continuous channels. Specifically, for channel input $\textbf{X}^n \in \mathbb{R}^n$, the channel output is
\begin{align}
\textbf{Y}^n=\textbf{X}^n+\textbf{Z}^n
\end{align}
where $\textbf{Z}^n\sim\mathcal{N}(0,\sigma_Z^2\textbf{I}_n)$ and  $\textbf{Y}^n \in \mathbb{R}^n$. In the context of fixed blocklength codes, the literature often defines various power constraints, which differ based on whether the restriction applies to each codeword or the entire codebook \cite{6767457, poly}. For VLSF codes, as defined in this paper, the codebook consists of infinite-dimensional codewords, and the most natural power constraint is as follows. Let the codeword $\textbf{c}_i \in \mathbb{R}^n$ that satisfies
\begin{align}
\lim_{n\rightarrow\infty}\frac{\Vert \textbf{c}_i^n \Vert^2}{n}=\sigma_X^2
\end{align}
for $i=1,...,M$. We use the notation of $\sigma_X^2$ for the power of the codeword because the power of a zero-mean signal is equal to its variance. Hence, the signal-to-noise ratio is $\gamma={\sigma_X^2}/{\sigma_Z^2}$.

The achievability bound for VLSF coding over the Gaussian channel can be assessed with the use of the following theorem.
\begin{theorem}
For the Gaussian channel with signal-to-noise ratio $\gamma$, let a channel output $\textbf{Y}$ given that the channel input is $\textbf{X}=(X_1,X_2,...,X_n,...)$. The channel input follows the $\mathcal{N}(\boldsymbol 0,\textbf{I}_\infty)$. Then

\begin{align}
\Lambda &= \Vert \textbf{Y}^n \Vert^2\\
d(\textbf{X}^n,\textbf{Y}^n) &= (X_1-Y_1)^2+...+(X_n-Y_n)^2,\\
\tau &= \inf\{n \geq 0:\lambda(d(\textbf{X}^n,\textbf{Y}^n),\Lambda)\leq \epsilon \},\\
\lambda(\upsilon,\nu) &= \frac{\lambda_1(\upsilon,\nu)}{\lambda_1(\upsilon,\nu)+\lambda_2(\upsilon,\nu)},
\end{align}
%
where
\begin{align}
\begin{split}
\lambda_1(\upsilon,\nu)= (M-1)&f_{\mathcal{X}_n^2}(\upsilon;\nu)\\
&\cdot\int_{\upsilon}^\infty  f_{\Gamma}(x;2^{-1}n,2\gamma^{-1})f_{\mathcal{X}_n^2}(\nu;x)dx,
\end{split}
\end{align}
and
\begin{align}
\lambda_2(\upsilon,\nu)=f_{\Gamma}(\upsilon;2^{-1}n,2\gamma^{-1})f_{\mathcal{X}_n^2}(\nu;\upsilon)(1-F_{\mathcal{X}_n^2}(\upsilon;\nu)).
\end{align}
\label{theor2}
Then for any $M$ there exists an $(l,M,\epsilon)$ VLSF code with
\begin{align}
\begin{split}
l &\leq \mathbb{E}[\tau].
\end{split}
\end{align}
\end{theorem}

\begin{proof}
We apply Theorem \ref{theorG}. The appropriate distance metric over $\mathbb{R}^n$ is the Euclidean distance
\begin{align}
d_E(\textbf{X}^n,\textbf{Y}^n) = \sqrt{(X_1-Y_1)^2+...+(X_n-Y_n)^2}.
\end{align}
However, as the codeword distances are only compared in search of the minimum, using the squared Euclidean distance is convenient,
\begin{align}
d(\textbf{X}^n,\textbf{Y}^n) = (X_1-Y_1)^2+...+(X_n-Y_n)^2.
\end{align}
Many distributions that describe this metric for random vectors are readily available, and the resulting bounds remain the same.

Each codeword is a standard normal random vector meaning that it is  a realization of the infinite-dimensional multivariate normal distribution with mean vector $\boldsymbol0$ and covariance matrix $\textbf{I}_{\infty}$. The choice of this distribution is based on the fact that it is capacity-achieving for the Gaussian channel, i.e. optimal for $n\rightarrow\infty$. We conjecture that it gives good bounds for finite $n$ as well.

Given that the codebook is distributed as mentioned above, $V_{1,n}=d(\textbf{X}_1,\textbf{Y})$ follows the the gamma distribution with shape $2^{-1}n$ and scale $2\gamma^{-1}$. The distances $d(\textbf{X}_i,\textbf{Y})$ for $i=2,...,M$ follow the non-central chi-squared distribution with non-centrality parameter
\begin{align}
\Lambda =\Vert \textbf{Y}^n \Vert^2.
\end{align}
Therefore, the cumulative distribution function of $V_{2,n}$ given $\Lambda = \nu$ is the following.
\begin{align}
\begin{split}
F_{V_{2,n}}(x;\nu)&=P(V_{2,n}\leq x)\\
&=1-P(V_{2,n}\geq x)\\
&=1-\prod_{i=2}^{M}P(d(\textbf{X}_i,\textbf{Y})\geq x)\\
&=1-P(d(\textbf{X}_2,\textbf{Y})\geq x)^{M-1}\\
&=1-(1-F_{\mathcal{X}_n^2}(x;\nu))^{M-1}.
\end{split}
\label{mincdf}
\end{align}
Then, the probability density function is 
\begin{align}
\begin{split}
f_{V_{2,n}}&(x;\nu)\\
&=\frac{dF_{V_{2,n}}(x;\nu)}{dx}\\
&=(M-1)(1-F_{\mathcal{X}_n^2}(x;\nu))^{(M-2)}f_{\mathcal{X}_n^2}(x;\nu).
\end{split}
\end{align}

Hence, the probability of error of the minimum distance decoder after the reception of the $n$th symbol can be evaluated by the function $\lambda(\upsilon,\nu)$ derived in (\ref{makrinar}) at the bottom of the page. In this derivation, $(a)$ is a result of the conditional independence of $V_{1,n}$ and $V_{2,n}$ given $\textbf{Y}^n$. Equation $(b)$ is a result of the spherical symmetry of the distributions of the channel input and output, where $\nu=\Vert \textbf{y}^n \Vert^2$. This step is further discussed in Appendix \ref{apendaB}.

\begin{figure*}[b]
\rule[1ex]{\textwidth}{0.1pt}
\begin{align}
\begin{split}
P(\hat{W} \neq W|\textbf{Y}^n=\textbf{y}^n) &=P(V_{2,n} = \upsilon|K_n =\upsilon,\textbf{Y}^n=\textbf{y}^n)\\
&= \frac{P(K_n = \upsilon|V_{2,n} = \upsilon,\textbf{Y}^n=\textbf{y}^n)P(V_{2,n} = \upsilon|\textbf{Y}^n=\textbf{y}^n)}{P(K_n = \upsilon|\textbf{Y}^n=\textbf{y}^n)}\\
&=\frac{P(V_{1,n} \geq \upsilon|V_{2,n} = \upsilon,\textbf{Y}^n=\textbf{y}^n)P(V_{2} = \upsilon|\textbf{Y}^n=\textbf{y}^n)}{P(V_{1,n} = \upsilon,V_{2,n} > \upsilon|\textbf{Y}^n=\textbf{y}^n)+P(V_{1,n} \geq \upsilon,V_{2} = \upsilon|\textbf{Y}^n=\textbf{y}^n)}\\
&\stackrel{(a)}{=}\frac{P(V_{1,n} \geq \upsilon|\textbf{Y}^n=\textbf{y}^n)P(V_{2,n} = \upsilon|\textbf{Y}^n=\textbf{y}^n)}{P(V_{1,n} = \upsilon|\textbf{Y}^n=\textbf{y}^n)P(V_{2,n} > \upsilon|\textbf{Y}^n=\textbf{y}^n)+P(V_{1,n} \geq \upsilon|\textbf{Y}^n=\textbf{y}^n)P(V_{2,n} = \upsilon|\textbf{Y}^n=\textbf{y}^n)}\\
&\stackrel{(b)}{=}\frac{P(V_{1,n} \geq \upsilon|\Lambda=\nu)P(V_{2,n} = \upsilon|\Lambda=\nu)}{P(V_{1,n} = \upsilon|\Lambda=\nu)P(V_{2,n} > \upsilon|\Lambda=\nu)+P(V_{1,n} \geq \upsilon|\Lambda=\nu)P(V_{2,n} = \upsilon|\Lambda=\nu)}\\
&=\frac{P(V_{1,n} \geq \upsilon,\Lambda=\nu)P(V_{2,n} = \upsilon|\Lambda=\nu)}{P(V_{1,n} = \upsilon,\Lambda=\nu)P(V_{2,n} > \upsilon|\Lambda=\nu)+P(V_{1,n} \geq \upsilon,\Lambda=\nu)P(V_{2,n} = \upsilon|\Lambda=\nu)}\\
&= \frac{f_{V_2}(\upsilon;n,\nu)\int_\upsilon^\infty f_{\Gamma}(x;2^{-1}n,2\gamma^{-1})f_{\mathcal{X}_n^2}(\nu;x)dx}{f_{\Gamma}(\upsilon;2^{-1}n,2\gamma^{-1})f_{\mathcal{X}_n^2}(\nu;\upsilon)(1-F_{V_2}(\upsilon;n,\nu))+f_{V_2}(\upsilon;n,\nu)\int_\upsilon^\infty f_{\Gamma}(x;2^{-1}n,2\gamma^{-1})f_{\mathcal{X}_n^2}(\nu;x)dx}\\
&=\frac{(M-1)(1-F_{\mathcal{X}_n^2}(\upsilon;\nu))^{(M-2)}f_{\mathcal{X}_n^2}(\upsilon;\nu)\int_{\upsilon}^\infty f_{\Gamma}(x;2^{-1}n,2\gamma^{-1})f_{\mathcal{X}_n^2}(\nu;x)dx}
{\splitfrac{\big[f_{\Gamma}(\upsilon;2^{-1}n,2\gamma^{-1})f_{\mathcal{X}_n^2}(\nu;\upsilon)(1-F_{\mathcal{X}_n^2}(\upsilon;\nu))^{M-1}}{+(M-1)(1-F_{\mathcal{X}_n^2}(\upsilon;\nu))^{(M-2)}f_{\mathcal{X}_n^2}(\upsilon;\nu)\int_{\upsilon}^\infty f_{\Gamma}(x;2^{-1}n,2\gamma^{-1})f_{\mathcal{X}_n^2}(\nu;x)dx\big]}}\\
&=\frac{(M-1)f_{\mathcal{X}_n^2}(\upsilon;\nu)\int_{\upsilon}^\infty f_{\Gamma}(x;2^{-1}n,2\gamma^{-1})f_{\mathcal{X}_n^2}(\nu;x)dx}{{f_{\Gamma}(\upsilon;2^{-1}n,2\gamma^{-1})f_{\mathcal{X}_n^2}(\nu;\upsilon)(1-F_{\mathcal{X}_n^2}(\upsilon;\nu))}{+(M-1)f_{\mathcal{X}_n^2}(\upsilon;\nu)\int_{\upsilon}^\infty f_{\Gamma}(x;2^{-1}n,2\gamma^{-1})f_{\mathcal{X}_n^2}(\nu;x)dx}}\\
&=\lambda(\upsilon,\nu)
\end{split}
\label{makrinar}
\end{align}
\end{figure*}

Following Theorem 1, define a stopping time 
\begin{align}
\bar{\tau} = \inf\{n \geq 0:\lambda(K_n,\Lambda)\leq \epsilon \}
\end{align}
where
\begin{align}
K_n = \min_{i=1,2}(V_{i,n}).
\end{align}
We can avoid this minimization and use as input $V_{1,n}$ to the function $\lambda$ since, as explained in \cite{6767457}, the Gaussian density is monotone with distance. This means that a higher distance results in a higher probability of error and a slightly relaxed upper bound in our case. Let a stopping time
\begin{align}
\tau = \inf\{n \geq 0:\lambda(V_{1,n},\Lambda)\leq \epsilon \}.
\end{align}
Since,
\begin{align}
\lambda(K_n,\Lambda)\leq \lambda(V_{1,n},\Lambda)
\end{align}
then
\begin{align}
l \leq \mathbb{E}[\bar{\tau}]\leq \mathbb{E}[\tau].
\end{align}

When $\epsilon$ is fixed near zero, this relaxation is negligible because the decoder stops the transmission anyway when  $V_{1,n}$ is less than  $V_{2,n}$ with probability $1-\epsilon$. Nevertheless, for the purpose of completeness, we discuss in Appendix \ref{apendA} how to efficiently sample realizations of $V_{2,n}$.
\end{proof}

\section{Numerical Analysis}

Theorem \ref{theor2} does not provide a closed-form formula for evaluation of the achievable rate of an $(l, M,\epsilon)$ VLSF code. However, it provides an upper bound to the probability of error of a single transmission and therefore can be used with Monte Carlo experiments to approximate $\mathbb{E}[\tau]$. In Figure \ref{f2}, we present the numerical evaluation of Theorem \ref{theor2} for signal-to-noise ratio $\gamma=1$ (0 dB) and average probability of error $\epsilon=10^{-3}$. Additionally, the normal approximation of fixed-length coding without feedback is presented for comparison purposes \cite[(296)]{poly}. Evidently, the significant improvements achieved by VLSF codes in \cite{5961844} for discrete channels are also attainable for the Gaussian channel.
 
\begin{figure*}
\centering
\includegraphics[scale=1]{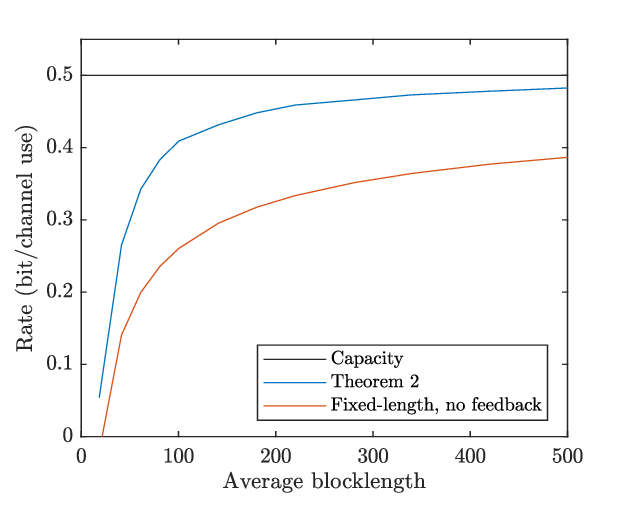}
\caption{Numerical evaluation of Theorem \ref{theor2} for signal-to-noise ratio $\gamma=1$ (0 dB) and average probability of error $\epsilon=10^{-3}$. The capacity of the channel and the normal approximation for fixed-length coding without feedback are also presented.} 
\label{f2}
\end{figure*}

\section{Conclusions}
\label{concl}

This paper provides an achievability bound on the rate of variable-length stop-feedback codes based on minimum distance decoding. This result can be used with any channel where minimum distance decoding is optimal, such as the BSC, the BEC, and the Gaussian channel. The particularization of this general bound for the Gaussian channel establishes the framework for approximating the average blocklength as well as the probability mass function of the decoding times via Monte Carlo experiments. Numerical analysis verifies that VLSF codes over the Gaussian channel have a significantly higher achievable rate for a given average blocklength compared to fixed-length codes without feedback. Our analysis assumes that feedback is instantaneous and it does not affect the rate of the code. An interesting question arises of how to utilize the probability mass function of the decoding times to select the optimal ones in order to maximize the rate $k/(n+n_f)$, where $k$ is the payload and $n_f$ is the number of the uses of the feedback channel.

\section*{Acknowledgment}
This work is supported by the Engineering and Physical Sciences Research Council (EP/L016656/1) and the University of Bristol.

\bibliography{vlsf_gaussian}

\begin{thebibliography}{1}

\bibitem{1056798}
C.~{Shannon}, ``The zero error capacity of a noisy channel,'' {\em IRE
  Transactions on Information Theory}, vol.~2, no.~3, pp.~8--19, 1956.

\bibitem{Bur76}
M.~Burnashev, ``Data transmission over a discrete channel with feedback.
  {R}andom transmission time,'' {\em Probl. Peredachi Inf.}, vol.~12, no.~4,
  pp.~10--30, 1976.

\bibitem{poly}
Y.~{Polyanskiy}, H.~V. {Poor}, and S.~{Verdu}, ``Channel coding rate in the
  finite blocklength regime,'' {\em IEEE Transactions on Information Theory},
  vol.~56, pp.~2307--2359, May 2010.

\bibitem{5961844}
Y.~{Polyanskiy}, H.~V. {Poor}, and S.~{Verdu}, ``Feedback in the non-asymptotic
  regime,'' {\em IEEE Transactions on Information Theory}, vol.~57, no.~8,
  pp.~4903--4925, 2011.

\bibitem{9204378}
J.~{Östman}, R.~{Devassy}, G.~{Durisi}, and E.~G. {Ström}, ``Short-packet
  transmission via variable-length codes in the presence of noisy stop
  feedback,'' {\em IEEE Transactions on Wireless Communications}, pp.~1--1,
  2020.

\bibitem{8259002}
L.~V. Truong and V.~Y.~F. Tan, ``On gaussian macs with variable-length feedback
  and non-vanishing error probabilities,'' {\em IEEE Transactions on
  Information Theory}, vol.~64, no.~4, pp.~2333--2346, 2018.

\bibitem{10206565}
I.~Papoutsidakis, A.~Doufexi, and R.~J. Piechocki, ``Efficient evaluation of
  the probability of error of random coding ensembles,'' in {\em 2023 IEEE
  International Symposium on Information Theory (ISIT)}, pp.~2117--2122, 2023.

\bibitem{6767457}
C.~E. Shannon, ``Probability of error for optimal codes in a gaussian
  channel,'' {\em The Bell System Technical Journal}, vol.~38, no.~3,
  pp.~611--656, 1959.

\bibitem{Devroye_1986}
L.~Devroye, {\em Non-Uniform Random Variate Generation}.
\newblock Springer New York, 1986.

\end{thebibliography}

\bibliographystyle{ieeetr}

\newpage
\appendices
\section{Complement to the proof of (\ref{makrinar})}
\label{apendaB}

In derivation (\ref{makrinar2}), at the bottom of the page, equality $(a)$ stems from the fact that $V_{2,n}$ is the minimum of non-central chi-squared random variables that depend only on $\Lambda=\Vert \textbf{y}^n \Vert^2=\nu$. To establish equality $(b)$ in (\ref{makrinar}), it is sufficient to demonstrate the following.

\begin{figure*}[b]
\rule[1ex]{\textwidth}{0.1pt}
\begin{align}
\begin{split}
P(\hat{W} \neq W|\textbf{Y}^n=\textbf{y}^n)&=\frac{P(V_{1,n} \geq \upsilon|\textbf{Y}^n=\textbf{y}^n)P(V_{2,n} = \upsilon|\textbf{Y}^n=\textbf{y}^n)}{P(V_{1,n} = \upsilon|\textbf{Y}^n=\textbf{y}^n)P(V_{2,n} > \upsilon|\textbf{Y}^n=\textbf{y}^n)+P(V_{1,n} \geq \upsilon|\textbf{Y}^n=\textbf{y}^n)P(V_{2,n} = \upsilon|\textbf{Y}^n=\textbf{y}^n)}\\
&\stackrel{(a)}{=}\frac{P(V_{1,n} \geq \upsilon|\textbf{Y}^n=\textbf{y}^n)P(V_{2,n} = \upsilon|\Lambda=\nu)}{P(V_{1,n} = \upsilon|\textbf{Y}^n=\textbf{y}^n)P(V_{2,n} > \upsilon|\Lambda=\nu)+P(V_{1,n} \geq \upsilon|\textbf{Y}^n=\textbf{y}^n)P(V_{2,n} = \upsilon|\Lambda=\nu)}\\
&=\frac{P(V_{2,n} = \upsilon|\Lambda=\nu)}{\frac{P(V_{1,n} = \upsilon|\textbf{Y}^n=\textbf{y}^n)}{P(V_{1,n} \geq \upsilon|\textbf{Y}^n=\textbf{y}^n)}P(V_{2,n} > \upsilon|\Lambda=\nu)+P(V_{2,n} = \upsilon|\Lambda=\nu)}
\end{split}
\label{makrinar2}
\end{align}
\end{figure*}

\begin{align}
\begin{split}
&\frac{P(V_{1,n} = \upsilon|\Lambda=\nu)}{P(V_{1,n} \geq \upsilon|\Lambda=\nu)}\\
&\stackrel{(a)}{=}\frac{\int_\textbf{x}P(V_{1,n} = \upsilon|\Lambda=\nu,\textbf{Y}^n=\textbf{x})f_{\textbf{Y}^n|\Lambda=\nu}(\textbf{x})d\textbf{x}}{\int_\textbf{x}P(V_{1,n} \geq \upsilon|\Lambda=\nu,\textbf{Y}^n=\textbf{x})f_{\textbf{Y}^n|\Lambda=\nu}(\textbf{x})d\textbf{x}}\\
&\stackrel{(b)}{=}\frac{\int_\textbf{x}P(V_{1,n} = \upsilon|\textbf{Y}^n=\textbf{x})f_{\textbf{Y}^n|\Lambda=\nu}(\textbf{x})d\textbf{x}}{\int_\textbf{x}P(V_{1,n} \geq \upsilon|\textbf{Y}^n=\textbf{x})f_{\textbf{Y}^n|\Lambda=\nu}(\textbf{x})d\textbf{x}}\\
&\stackrel{(c)}{=}\frac{\int_\textbf{x}P(V_{1,n} = \upsilon|\textbf{Y}^n=\textbf{y}^n)f_{\textbf{Y}^n|\Lambda=\nu}(\textbf{x})d\textbf{x}}{\int_\textbf{x}P(V_{1,n} \geq \upsilon|\textbf{Y}^n=\textbf{y}^n)f_{\textbf{Y}^n|\Lambda=\nu}(\textbf{x})d\textbf{x}}\\
&=\frac{P(V_{1,n} = \upsilon|\textbf{Y}^n=\textbf{y}^n)\int_\textbf{x}f_{\textbf{Y}^n|\Lambda=\nu}(\textbf{x})d\textbf{x}}{P(V_{1,n} \geq \upsilon|\textbf{Y}^n=\textbf{y}^n)\int_\textbf{x}f_{\textbf{Y}^n|\Lambda=\nu}(\textbf{x})d\textbf{x}}\\
&=\frac{P(V_{1,n} = \upsilon|\textbf{Y}^n=\textbf{y}^n)}{P(V_{1,n} \geq \upsilon|\textbf{Y}^n=\textbf{y}^n)}.
\end{split}
\end{align}
Here, $(a)$ follows from the law of total probability, $(b)$ from noting that the integrand is greater than zero only when $\Vert \textbf{x} \Vert^2 = \Vert \textbf{y}^n \Vert^2 =\nu$; in this case, the condition $\Lambda=\nu$ does not provide extra information. Finally, $(c)$ arises from the spherical symmetry of the distributions of the channel input, the noise, and the channel output.  This implies that
\begin{align}
P(V_{1,n} = \upsilon|\textbf{Y}^n=\textbf{y}^n)=P(V_{1,n} = \upsilon|\textbf{Y}^n=\bar{\textbf{y}}^n)
\end{align}
for any $\bar{\textbf{y}}^n$ such that $\Vert \bar{\textbf{y}}^n \Vert^2=\Vert \textbf{y}^n \Vert^2$.
\section{Sampling of $V_{2,n}$}
\label{apendA}
As mentioned in Section \ref{gaus_res}, the production of samples of $V_{2,n}$ can be avoided with a slight relaxation of the achievability bound. Despite this fact, we give a method based on inverse sampling \cite[Theorem 2.1]{Devroye_1986}. Its basic concept is the following. Let an continuous random variable $T$ with cdf $F_T$ and quantile function $F_T^{-1}$. Then,
\begin{align}
T = F_T^{-1}(U)
\end{align}
where $U$ is uniformly distributed in $[0,1]$.

A recursive method is proposed for the sampling of $V_{2,n}$. Assuming, without loss of generality, that the realization of $d(\textbf{X}^{n-1}_2,\textbf{Y}^{n-1})=k$ is the minimum, hence
\begin{align}
V_{2,n-1} = k.
\end{align}
Note that
\begin{align}
\begin{split}
V_{2,n} &= \min_{j=2,...,M} d(\textbf{X}^{n}_j,\textbf{Y}^{n})\\
&=V_{2,n-1}+\min_{j=2,...,M} (d(\textbf{X}^{n}_j,\textbf{Y}^{n})-V_{2,n-1})\\
&=V_{2,n-1}+\min(Z_1,Z_2).
\end{split}
\label{rngidea}
\end{align}
where $Z_1$ follows the non-central chi-squared distribution with $1$ degree of freedom and non-centrality parameter $Y_n^2$ and
\begin{align}
\begin{split}
Z_2 &=  \min_{j=3,...,M} (d(\textbf{X}^{n}_j,\textbf{Y}^{n})-V_{2,n-1})\\
&= \min_{j=3,...,M} Q_j.
\end{split}
\end{align}
Random variables $Q_j$ for $j=3,...,M$ are conditionally independent given $\textbf{Y}^n$ and $V_{2,n-1}$, and they are identically distributed as follows.
\begin{align}
Q_j \sim W = W_1+W_2.
\end{align}
Random variable $W_1$ follows the same distribution as $Z_1$ and $W_2$ follows the conditioned non-central chi-squared distribution with pdf 
\begin{align}
\begin{split}
f_{W_2}(x)&=f_{\mathcal{X}_{n-1}^2|\mathcal{X}_{n-1}^2\geq k}(x;\eta)\\
&=\frac{1\{x\geq k\}f_{\mathcal{X}_{n-1}^2}(x;\eta)}{1-F_{\mathcal{X}_{n-1}^2}(k;\eta)}
\end{split}
\end{align}
where
\begin{align}
\eta = \Vert \textbf{y}^{n-1} \Vert^2.
\end{align}
Then, the cdf of W is

\begin{align}
\begin{split}
F_W(w)&=\int_0^\infty F_{W_1}(w-x)f_{W_2}(x)dx \\
&=\int_0^\infty F_{W_1}(w-x)\frac{1\{x\geq k\}f_{\mathcal{X}_{n-1}^2}(x;\eta)}{1-F_{\mathcal{X}_{n-1}^2}(k;\eta)}dx\\
&=\int_{k}^\infty F_{W_1}(w-x)\frac{f_{\mathcal{X}_{n-1}^2}(x;\eta)}{1-F_{\mathcal{X}_{n-1}^2}(k;\eta)} dx\\
&=\int_{0}^\infty F_{W_1}(w-x)\frac{f_{\mathcal{X}_{n-1}^2}(x;\eta)}{1-F_{\mathcal{X}_{n-1}^2}(k;\eta)}dx\\
&\quad-\int_0^{k} F_{W_1}(w-x)\frac{f_{\mathcal{X}_{n-1}^2}(x;\eta)}{1-F_{\mathcal{X}_{n-1}^2}(k;\eta)}dx\\
&=\frac{\splitfrac{\big[\int_{0}^\infty F_{W_1}(w-x)f_{\mathcal{X}_{n-1}^2}(x;\eta)dx}{-\int_0^{k} F_{W_1}(w-x)f_{X_{n-1}^2}(x;\eta)dx)\big]}}{1-F_{\mathcal{X}_{n-1}^2}(k;\eta)}\\
&=\frac{{F_{\mathcal{X}_{n}^2}(w;\Vert \textbf{y}^{n} \Vert^2)}{-\int_0^{k} F_{W_1}(w-x)f_{\mathcal{X}_{n-1}^2}(x;\eta)dx}}{1-F_{\mathcal{X}_{n-1}^2}(k;\eta)}.
\end{split}
\end{align}
Finally, similarly to (\ref{mincdf}), the cdf of $Z_2$ is
\begin{align}
F_{Z_2}(x) =1-(1-F_{W}(x))^{M-2}.
\end{align}
Since $F_{Z_2}^{-1}$ is a monotonic function, it can be numerically approximated in an efficient manner. When combined with the recursion (\ref{rngidea}), it can then be used to generate realizations of $V_{2,n}$.

\end{document}